\documentclass[11pt]{article} 
\usepackage[a4paper, margin=1in]{geometry}
\usepackage[T1]{fontenc}
\usepackage[utf8]{inputenc}
\usepackage{lmodern}
\usepackage{microtype}
\usepackage{amsmath, amssymb, amsthm, mathrsfs}
\usepackage{graphicx}
\usepackage{pgfplots}
\usepackage{tikz}
\usetikzlibrary{shapes.geometric, arrows, positioning}
\usepackage[numbers]{natbib}
\usepackage[colorlinks=true, linkcolor=blue, citecolor=blue, urlcolor=blue]{hyperref}
\hypersetup{
    pdftitle={Symbolic Generation and Modular Embedding of High-Quality abc-Triples},
    pdfauthor={Michael A. Idowu},
    pdfsubject={Number Theory, Symbolic Computation, Cryptography},
    pdfkeywords={abc conjecture, modular inversion, affine transformation, radical minimisation, symbolic computation, Diophantine analysis}
}

\usepackage{url} 
\usepackage{hyperref} 

\usepackage{enumitem}
\usepackage{booktabs} 
\usepackage[british]{babel} 

\usepackage{tikz}

\tikzstyle{block} = [rectangle, rounded corners, minimum width=3.5cm, minimum height=1cm, text centered, draw=black, fill=white]
\tikzstyle{arrow} = [thick, ->, >=stealth]

\title{Symbolic Generation and Modular Embedding of High-Quality \textit{abc}-Triples}

\author{Michael A. Idowu\thanks{Email: \texttt{michade@hotmail.com}}}

\date{}

\newtheorem{theorem}{Theorem}[section]
\newtheorem{lemma}[theorem]{Lemma}
\newtheorem{proposition}[theorem]{Proposition}
\newtheorem{definition}[theorem]{Definition}

\begin{document}

\maketitle

\begin{abstract}
We present a symbolic identity for generating integer triples \((a, b, c)\) satisfying \(a + b = c\), inspired by structural features of the \textit{abc} conjecture. The construction uses powers of 2 and 3 in combination with modular inversion in \(\mathbb{Z}/3^p\mathbb{Z}\), leading to a parametric identity with residue constraints that yield \textit{abc}-triples exhibiting low radical values. Through affine transformations, these symbolic triples are embedded into a broader space of high-quality examples, optimised for the ratio \(\log c / \log \mathrm{rad}(abc)\). Computational results demonstrate the emergence of structured, radical-minimising candidates, including both known and novel triples. These methods provide a symbolic and algebraic framework for controlled triple generation, and suggest exploratory implications for symbolic entropy filtering in cryptographic pre-processing.
\end{abstract}


\section{Introduction}
The \textit{abc} conjecture, independently proposed by Masser and Oesterlé in the 1980s \cite{Masser1985, Oesterle1988}, concerns triples of coprime positive integers \( (a, b, c) \) such that
\[
a + b = c,
\]
and posits that for every \( \varepsilon > 0 \), there exist only finitely many such triples with
\[
c > \operatorname{rad}(abc)^{1+\varepsilon},
\]
where \( \operatorname{rad}(n) \) denotes the product of the distinct prime factors of \( n \). The associated \emph{quality} of such a triple is defined as
\begin{equation}
q(a, b, c) = \frac{\log c}{\log \operatorname{rad}(abc)}.
\end{equation}
The conjecture implies that \( q(a, b, c) < 1 + \varepsilon \) for almost all such triples and suggests a deep connection between the additive and multiplicative structures of positive integers, which has far-reaching consequences in number theory, impacting Diophantine equations, elliptic curves, and arithmetic geometry\cite{Silverman2007, vanFrankenhuijsen2006}.

In this work, we propose a symbolic and modular identity to generate structured additive triples \((a, b, c)\) with characteristics that mirror those involved in the \textit{abc} conjecture. The construction employs powers of 2 and 3, coupled with modular inversion in \(\mathbb{Z}/3^p\mathbb{Z}\), to define residue-constrained identities through congruence conditions on auxiliary parameters.

We analytically derive a modular inversion constraint that determines admissible values and explore its theoretical and computational implications. Using affine transformations, these symbolic triples are embedded into a larger space of known high-quality \textit{abc}-triples, supporting symbolic encoding, transformation, and radical optimisation. This approach blends modular arithmetic and inverse cycles, offering a novel algebraic perspective on triple structures.

\subsection{Main Contributions}
Our key contributions include:
\begin{itemize}[leftmargin=2em]
    \item A parametric identity that generates additive triples \((a, b, c)\) satisfying \(a + b = c\), constrained by modular inversion in \(\mathbb{Z}/3^p\mathbb{Z}\).
    \item A residue class condition that characterises admissible values of the symbolic parameter \(d\).
    \item An affine transformation framework for symbolic triples, aligning them with classical high-quality \textit{abc}-triples while enabling radical minimisation.
    \item Computational results demonstrating the identity's effectiveness in producing low-radical triples with cryptographic relevance.
    \item Exploratory cryptographic discussion, focusing on entropy filtering and modular predictability.
\end{itemize}

\textit{Remark.} In addition to computational and symbolic methods, it is worth noting the controversial proof of the \textit{abc} conjecture claimed by Shinichi Mochizuki via Inter-universal Teichmüller (IUT) theory. While his work remains under active scrutiny and has not achieved broad consensus, it represents a significant development in the theoretical discourse surrounding the conjecture; see \cite{Mochizuki2020} for details.

\subsection{Context within Existing \textit{abc}-Triple Generation Work}
Traditional searches for high-quality \textit{abc}-triples rely on enumeration, smoothness filtering, or probabilistic sieving\cite{Browkin1994, Granville1998}. These methods explore coprime additive triples ranked by
\[ Q(a, b, c) = \frac{\log c}{\log \mathrm{rad}(abc)}. \]
Advanced approaches, such as ABC@home or the work of Nitaj \cite{Nitaj2019}, incorporate distributed computing and heuristic bounds informed by number theory. Symbolic constructions remain less prevalent.\\
Our method diverges by introducing an algebraically motivated identity underpinned by residue constraints and affine embeddings. It provides:
\begin{itemize}[leftmargin=2em]
    \item A symbolic alternative to brute-force or statistical candidate generation.
    \item Analytically bounded parameters for algebraic post-processing.
    \item A tunable symbolic structure that favours radical minimisation and compressibility.
\end{itemize}

This contribution complements computational methods and opens up new paths for analysis, including extensions relevant to cryptographic applications. The paper proceeds as follows: Section~2 introduces the parametric identity and modular constraints. Section~3 presents the affine transformation model. Section~4 explores structural optimisation. Section~5 delivers computational evidence. Section~6 explores cryptographic applications. Section~7 concludes and outlines future directions.

\section*{Notation and Symbol Table}
\begin{table}[h!]
\centering
\caption{Summary of mathematical symbols used throughout the paper.}
\label{tab:symbols}
\renewcommand{\arraystretch}{1.2}
\begin{tabular}{|p{3cm}|p{10cm}|}
\hline
\textbf{Symbol} & \textbf{Meaning} \\
\hline
$a, b, c$ & Integer triple with $a + b = c$ \\
$p, k, n$ & Symbolic parameters: exponent indices \\
$d$ & Modular inverse constraint value \\
$\mathrm{rad}(n)$ & Product of distinct prime divisors of $n$ \\
$Q(a, b, c)$ & Triple quality: $\log c / \log \mathrm{rad}(abc)$ \\
$\eta(p, k, d)$ & Symbolic compression ratio \\
$\alpha, \gamma, \delta$ & Affine transformation parameters \\
\hline
\end{tabular}
\end{table}

\newpage

\section{Parametric Identity and Modular Constraints}
\label{sec:identity}
We consider the identity
\begin{equation}
3^p(s + 1) = 1 + 2^{k-1}(2 \cdot 3^p \cdot n + d),
\label{eq:main_identity}
\end{equation}
where \( d \) is odd, \( s \in 2\mathbb{Z} \), and \( p, k, n \in \mathbb{Z}^+ \). Our goal is to recast this identity into a form consistent with the \textit{abc} conjecture, explore parameter regimes that maximise the associated \emph{quality}, and investigate structural implications using the binomial expansion of \( 3^p \).\\
We define a symbolic identity for generating integer triples \((a, b, c)\) satisfying \(a + b = c\), with \(a = 1\) and \(b\) expressed via a modularly-constrained function:
\begin{align*}
a &= 1, \\
b &= 2^{k-1}(2 \cdot 3^p \cdot n + d), \\
c &= 3^p(s + 1) = a + b.
\end{align*}
This structure fits the \textit{abc} format \( a + b = c \), and we wish to identify parameter regimes in which the quality \( q(a, b, c) \) is maximised and \(p, k \in \mathbb{Z}_{>0}\) are parameters.

\subsection{Core Identity}
We define:
\[ a = 1, \quad b = 2^{k-1}(2 \cdot 3^p n + d), \quad c = a + b, \]
where \(n \in \mathbb{Z}_{\geq 0}\), and \(d\) is a symbolic parameter subject to a modular constraint derived from the additive identity.

%
%
%
%

\subsection{Modular Inversion Constraint}
We require that \(a + b = c\) with the goal of constructing symbolic triples with a particular structure. For the identity to be consistent with symbolic embedding and invertibility over \(\mathbb{Z}/3^p\mathbb{Z}\), we derive the following constraint:

\begin{lemma}[Modular Inverse Constraint]
Let \(k \in \mathbb{Z}_{>0}\), then
\[ d \equiv -\left(2^{k-1}\right)^{-1} \mod 3^p. \]
\end{lemma}

\begin{proof}
To ensure that the additive relation is preserved under symbolic inversion, we examine the inverse of \(2^{k-1}\) modulo \(3^p\). Since 2 is coprime to 3, \(2^{k-1}\) has an inverse in \(\mathbb{Z}/3^p\mathbb{Z}\). Thus, solving
\[ 2^{k-1} d \equiv -1 \mod 3^p \]
leads to
\[ d \equiv -\left(2^{k-1}\right)^{-1} \mod 3^p. \qedhere \]
\end{proof}

\subsubsection{Guiding Principle and Parameter Choices for Maximising Quality}

We aim to maximise \( \log c \), which grows with \( p \), and minimise \( \log \operatorname{rad}(abc) \) by controlling prime factors of \( b \) and \( c \):
\begin{itemize}[noitemsep]
    \item Set \( a = 1 \) to guarantee coprimality with other terms.
    \item Choose \( p \gg 1 \): exponentially increases \( c = 3^p(s+1) \).
    \item Choose \( k \) such that \( c \) is divisible by \( 3^p \), ensuring \( s \in \mathbb{Z} \).
    \item Determine values of \(n \) and \(d \) consistent with \( c \) to keep \( b \) small and reduce its prime factorisation.    
\end{itemize}

\subsubsection{Smoothness Condition}
To minimise the radical \( \operatorname{rad}(abc) \), we prefer that \( 2 \cdot 3^pn + d \) be \emph{smooth} (i.e., composed of small primes). This leads to small \( \operatorname{rad}(b) \), since:
\[
b = 2^{k-1}(2 \cdot 3^pn + d).
\]

\subsection{Illustrative Example: Computation of \( d \)}
Let \(p = 2\), \(k = 3\). Then \(3^p = 9\), \(2^{k-1} = 4\). Compute the modular inverse of 4 modulo 9:
\[ 4^{-1} \equiv 7 \mod 9, \quad \text{since } 4 \cdot 7 = 28 \equiv 1 \mod 9. \]
Thus,
\[ d \equiv -7 \equiv 2 \mod 9. \]

Then,
\[ b = 4(2 \cdot 3^2 \cdot n + 11) = 4(18n + 11). \]
For \(n = 0\), we obtain \(b = 44\), \(c = a + b = 3^2(4+1)\).

\subsection{Illustrative Example: Computation of \( k \)}

Let \( p = 5 \), \( s = 2 \), and find \( k \) such that \( c = 3^p(s + 1) \) is an integer:
\begin{align*}
3^p &= 243, \\
s &= 2 \\
b = 2^{k-1}(2 \cdot 3^pn + d) &= 728.
\end{align*}
\( k = 4 \Rightarrow d = 91 \Rightarrow c = 729 \). \\
Now compute:
\[
\operatorname{rad}(abc) = \operatorname{rad}(1 \cdot 728 \cdot 729).
\]
We factor:
\[
728 = 2^{4-1}(91)= 2^{3} \cdot 7 \cdot 13, \quad 729 = 3^6,
\]
hence:
\[
\operatorname{rad}(abc) = 2 \cdot 3 \cdot 7 \cdot 13  = 546.
\]

Calculate the quality:
\begin{align*}
\log c &\approx \log(729) \approx 2.8627, \\
\log \operatorname{rad}(abc) &\approx \log(546) \approx 2.737, \\
q &\approx \frac{2.8627}{2.737} \approx 1.04586.
\end{align*}

This already exceeds 1, and larger \( p \) values will improve \( \log c \) further.

\subsection{Asymptotic Behaviour}
Assume \( s + 1 \) remains bounded or grows slowly. Then:
\[
\log c = \log(3^p(s + 1)) = p \log 3 + \log(s + 1),
\]
and with \( \log b = (k-1)\log 2+ \log(2 \cdot 3^p n+ d) \), we have:
\begin{equation}
\begin{split}
\log c &> \log(b)   \\
\rightarrow p \log 3 + \log(s + 1)  &\approx (k-1)\log 2+ \log(2 \cdot 3^p n+ d),
\end{split}
\end{equation}
 as \( p \to \infty \).\\
But 
\begin{equation}
\begin{split}
\log \operatorname{rad}(abc) &= \log \operatorname{rad}(2 \cdot 3 \cdot (s+1) \cdot (2 \cdot 3^p n+ d)) \\
&\le \log (2 \cdot 3) +\log \operatorname{rad}((s+1) \cdot (2 \cdot 3^p n+ d)).
\end{split}
\end{equation}
Thus, as \( p \to \infty \),
\begin{equation}
\begin{split}
q(a, b, c) &= \frac{ p \log 3 + \log(s + 1) }{\log \operatorname{rad}(abc)} \\
&= \frac{ p \log 3 + \log(s + 1) }{\log \operatorname{rad}(2 \cdot 3 \cdot (s+1) \cdot (2 \cdot 3^p n+ d))} \\
& \sim \frac{(k-1)\log 2+ \log(2 \cdot 3^p n+ d)}{\log \operatorname{rad}(2 \cdot 3 \cdot (s+1) \cdot (2 \cdot 3^p n+ d))}. \\ 
\end{split}
\end{equation}


\subsubsection{Key Obstacle}

The term \( 2 \cdot 3^pn + d \) tends to be \emph{rough} (i.e., contains large prime factors) for large \( p \) and dependent on \( n \), reducing the likelihood that the radical stays bounded. Therefore, we expect only finitely many high-quality examples, consistent with the \textit{abc} conjecture. Understanding the relationship between smooth values of \( 2 \cdot 3^pn + d \) and all the other parameters and variables remains a critical challenge.
The identity in Eq.~\eqref{eq:main_identity} provides a constructive method for generating \textit{abc} triples with potentially high quality. The quality increases with \( p \) under controlled radical growth. However, the rarity of smooth values of \( 2 \cdot 3^pn + d \) limits the number of such high-quality cases.

%

\subsection{Revised and More Insightful Examples}

To better illustrate the construction of abc-triples, we refine the earlier examples by making the assumptions and computations more explicit, particularly ensuring that the parameter \(d\) is odd as required.

\subsubsection{Example 1: Computation of \(d\) (with \(d\) odd)}
\label{comp_d}
Let \(p = 2\), so \(3^p = 9\), and let \(k = 3\), so \(2^{k-1} = 4\). We compute the modular inverse of 4 modulo 9:

\[
4^{-1} \equiv 7 \mod 9, \quad \text{since } 4 \times 7 = 28 \equiv 1 \mod 9.
\]
Thus,
\[
d \equiv -4^{-1} \equiv -7 \equiv 2 \mod 9.
\]
We now choose \(d = 11\), which satisfies \(d \equiv 2 \mod 9\) and is odd, as required.\\
Define:
\[
b = 2^{k-1}(2 \cdot 3^p n + d) = 4(18n + 11).
\]
For \(n = 0\), we obtain:
\[
b = 4 \cdot 11 = 44, \quad a = 1, \quad c = a + b = 45.
\]
Check that \(c = 3^2(s + 1)\). Since \(3^2 = 9\), we find:
\[
s = \frac{c}{9} - 1 = \frac{45}{9} - 1 = 4.
\]
Now factor:
\[
abc = 1 \cdot 44 \cdot 45 = (2^2 \cdot 11)(3^2 \cdot 5) = 2^2 \cdot 3^2 \cdot 5 \cdot 11,
\]
\[
\operatorname{rad}(abc) = 2 \cdot 3 \cdot 5 \cdot 11 = 330.
\]
Then compute the quality:
\[
q = \frac{\log c}{\log \operatorname{rad}(abc)} = \frac{\log 45}{\log 330} \approx \frac{3.807}{5.799} \approx 0.656.
\]
Although this quality is below 1, it demonstrates the mechanics of the construction. Higher values of \(p\) will improve the quality.

\subsubsection{Example 2: Computation of \(k\)}
\label{comp_k}
Let \(p = 5\), so \(3^p = 243\), and take \(s = 2\), so:
\[
c = 3^p(s + 1) = 243 \cdot 3 = 729.
\]
Let \(k = 4\), so \(2^{k-1} | (c-1) \), and suppose:
\[
b = 2^{k-1}(2 \cdot 3^p n + d) = 2^{4-1}(2 \cdot 729 \cdot n + d).
\]
Therefore \(n = 0\) and \(d = 91\), which is odd.\\
Then:
\[
b = 8 \cdot 91 = 728, \quad a = 1, \quad c = a + b = 729.
\]
Factor the components:
\[
a = 1, \quad b = 2^3 \cdot 7 \cdot 13, \quad c = 3^6.
\]
Hence:
\[
\operatorname{rad}(abc) = \operatorname{rad}(1 \cdot 728 \cdot 729) = 2 \cdot 3 \cdot 7 \cdot 13 = 546.
\]
Now compute the logarithmic quality:
\[
\log c \approx \log(729) \approx 6.5917, \quad \log \operatorname{rad}(abc) \approx \log(546) \approx 6.303,
\]
\[
q \approx \frac{6.5917}{6.303} \approx 1.0459.
\]
This example yields a valid abc-triple with quality exceeding 1.\\
Both examples in \ref{comp_d} and \ref{comp_k} satisfy the abc-condition: \(a + b = c\), \(a > 0\), \(b > 0\), and \(\gcd(a, b)\).\\
The first example illustrates the computation of a valid \(d\) ensuring parity and congruence.\\
The second example demonstrates how increasing \(p\) allows for higher-quality values due to exponential growth in \(c\).

\subsection{High-Quality \textit{abc}-Triple Example and Parameterised Construction}

To construct a triple \( (a, b, c) \) satisfying the conditions of the \textit{abc}-conjecture, we seek:
 \( a + b = c \), \( \gcd(a, b) = 1 \), \( a, b, c > 0 \), \( \operatorname{rad}(abc) \ll c \), so that \( q(a, b, c) = \dfrac{\log c}{\log \operatorname{rad}(abc)} > 1 \).\\
We consider a parameterised form of the construction:
\[
c = 3^p(s + 1), \quad b = 2^{k-1}(2 \cdot 3^p n + d), \quad a = c - b = 1,
\]
where: \( d \) is an odd integer satisfying \( d \equiv -(2^{k-1})^{-1} \pmod{3^p} \),
 \( \gcd(a, b) = 1 \), \( d \) is selected to avoid introducing large prime divisors.

\subsubsection{Example 3: Computation of \( d \) (with \( Q > 1.4 \))}
\label{comp_d_highQ}

Let \( p = 5 \), so \( 3^p = 243 \), and let \( k = 13 \), giving \( 2^{k-1} = 2^{12} = 4096 \). We aim to compute the modular inverse of \( 2^{12} \) modulo \( 3^5 \), ensuring that the result is \textbf{odd}.

\paragraph{Step 1: Evaluate the relevant powers.}
\[
2^{12} = 4096, \qquad 3^5 = 243.
\]

\paragraph{Step 2: Apply the extended Euclidean algorithm.}
We solve:
\[
4096x + 243y = 1
\]
Proceeding with the Euclidean algorithm:
\begin{align*}
4096 &= 16 \cdot 243 + 208 \\
243 &= 1 \cdot 208 + 35 \\
208 &= 5 \cdot 35 + 33 \\
35 &= 1 \cdot 33 + 2 \\
33 &= 16 \cdot 2 + 1 \\
2 &= 2 \cdot 1 + 0
\end{align*}

Since \(\gcd(4096, 243) = 1\), an inverse exists.

\paragraph{Step 3: Perform back-substitution.}
Expressing 1 as a linear combination:
\begin{align*}
1 &= 33 - 16 \cdot 2 \\
  &= 33 - 16(35 - 1 \cdot 33) \\
  &= 17 \cdot 33 - 16 \cdot 35 \\
  &= 17(208 - 5 \cdot 35) - 16 \cdot 35 \\
  &= 17 \cdot 208 - 101 \cdot 35 \\
  &= 17 \cdot 208 - 101(243 - 1 \cdot 208) \\
  &= 118 \cdot 208 - 101 \cdot 243 \\
  &= 118(4096 - 16 \cdot 243) - 101 \cdot 243 \\
  &= 118 \cdot 4096 - (1888 + 101) \cdot 243 \\
  &= 118 \cdot 4096 - 1989 \cdot 243
\end{align*}
Thus,
\[
(2^{12})^{-1} \equiv 118 \pmod{243}.
\]

\paragraph{Step 4: Compute \( d \).}
\[
d \equiv - (2^{12})^{-1} \equiv -118 \equiv 125 \pmod{243}.
\]
We now choose \( d = 5^3 = 125 \), which satisfies this congruence and is odd, as required.

\paragraph{Step 5: Define the symbolic triple.}
Let:
\[
b = 2^{12}(2 \cdot 3^p n + d) = 4096(2 \cdot 243n + 125).
\]
For \( n = 0 \), we get:
\[
a = 1, \quad b = 4096 \cdot 125 = 512000, \quad c = a + b = 512001.
\]
Note that:
\[
c = 3^5 \cdot 49 \cdot 43.
\]

\paragraph{Step 6: Compute the radical of \( abc \).}
\[
abc = 1 \cdot (2^{12} \cdot 5^3) \cdot (3^5 \cdot 7^2 \cdot 43)
\]
Distinct prime divisors:
\[
\{2, 3, 5, 7, 43\} \Rightarrow \operatorname{rad}(abc) = 2 \cdot 3 \cdot 5 \cdot 7 \cdot 43 = 9030.
\]

\paragraph{Step 7: Compute the \textit{abc}-quality.}
\[
\text{quality}(a, b, c) = \frac{\log c}{\log \operatorname{rad}(abc)}
\]
Substituting:
\[
\log c \approx \log 512001 \approx 13.146, \quad \log \operatorname{rad}(abc) \approx \log 9030 \approx 9.108,
\]
\[
\Rightarrow \text{quality} \approx \frac{13.146}{9.108} \approx \boxed{1.4437}
\]

\medskip
This quality significantly exceeds 1.4, highlighting a high-quality \textit{abc}-triple with controlled radical growth due to carefully selected symbolic parameters. The exponential growth in \( c \), via powers of 3, provides a substantial numerator in the logarithmic ratio, reinforcing the efficiency of this construction.

\subsubsection{Construction Template}

\begin{enumerate}
  \item Choose \( p \), and set \( c = 3^p(s + 1) \).
  \item Choose \( k \), and find the modular inverse \( 2^{k-1}{}^{-1} \mod 3^p \).
  \item Set \( d \equiv -2^{k-1}{}^{-1} \mod 3^p \), ensuring \( d \) is odd.
  \item Define \( b = 2^{k-1}(2 \cdot 3^p n + d) \), and set \( a = c - b \).
  \item Verify: \( a > 0 \), \( \gcd(a, b) = 1 \); then compute \( \operatorname{rad}(abc) \) and the corresponding quality \( Q \).
\end{enumerate}


\section{Affine Transformation and Symbolic Embedding}
To relate the parametrically generated triples to known high-quality \textit{abc}-triples, we introduce an affine transformation framework. This allows us to view the identity-based triples as basis elements from which other triples may be derived via linear scaling and translation.

\begin{samepage}
\subsection{Symbolic Construction Workflow}
By fixing $k$ and varying $p$, one obtains a predictable orbit of inverses that may be used to seed cryptographic routines while avoiding repetition across independent key derivation stages \ref{fig:workflow}. These values may serve as inverse-consistent initialisation vectors in deterministic pseudorandom number generators (PRNGs), where cryptographic entropy must remain within a controllable and verifiable structure.

\begin{figure}[h]
\centering
\begin{tikzpicture}[node distance=1.2cm]
\tikzstyle{block} = [rectangle, rounded corners, draw, minimum height=2em, text centered]
\tikzstyle{arrow} = [thick,->,>=stealth]

\node (start) [block] {Symbolic Parameters $(p, k, d)$};
\node (modinv) [block, below of=start] {Compute $d \equiv - \left(2^{k-1} \right)^{-1} \mod 3^p$};
\node (triple) [block, below of=modinv] {Generate $(a, b, c)$};
\node (affine) [block, below of=triple] {Affine Embedding};
\node (optimise) [block, below of=affine] {Radical Minimisation};
\node (crypto) [block, below of=optimise] {Cryptographic Utility};

\draw [arrow] (start) -- (modinv);
\draw [arrow] (modinv) -- (triple);
\draw [arrow] (triple) -- (affine);
\draw [arrow] (affine) -- (optimise);
\draw [arrow] (optimise) -- (crypto);
\end{tikzpicture}
\caption{Pipeline from symbolic parameter selection to cryptographic application.}
\label{fig:workflow}
\end{figure}
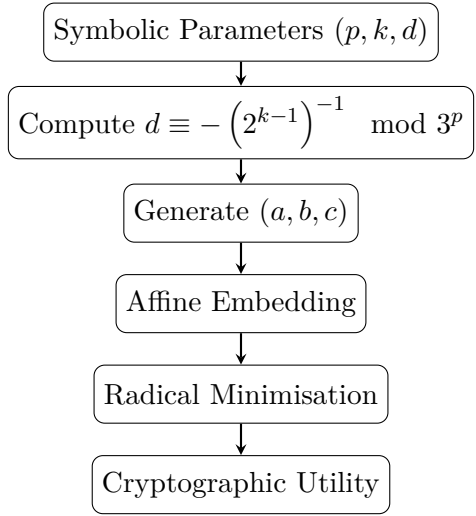
\end{samepage}

\subsection{Symbolic Triple}
\begin{definition}[Symbolic Triple]
Let \((a, b, c)\) be a triple generated from the identity:
\[a = 1, \quad b = 2^{k-1}(2 \cdot 3^p n + d), \quad c = a + b.
\]
We refer to \((a, b, c)\) as a symbolic triple, parameterised by \((p, k, d, n)\).
\end{definition}

\subsection{Affine Embedding Preserves Additivity}
\begin{theorem}[Affine Embedding Preserves Additivity]
Let \((a, b, c)\) be a symbolic triple as above. Define an affine transformation:
\[(a', b', c') = \alpha(a, b, c) + (\gamma, \delta, \beta), \quad \alpha \in \mathbb{Z}_{>0},
\]
Then the additive relation \(a' + b' = c'\) holds if and only if \(\gamma + \delta = \beta\).
\end{theorem}

\begin{proof}
Compute:
\[a' + b' = \alpha(a + b) + \gamma + \delta = \alpha c + \gamma + \delta, \quad c' = \alpha c + \beta.\]
Thus, \(a' + b' = c'\) iff \(\gamma + \delta = \beta\). \qedhere
\end{proof}

%
%
%

\subsection{Radical Control Under Affine Transformation}
\begin{proposition}[Radical Control Under Affine Transformation]
Let \((a, b, c)\) be a symbolic triple and \((a', b', c') = \alpha(a, b, c) + (\gamma, \delta, \beta)\) its affine image with \(\gamma + \delta = \beta\). If \(\gcd(\alpha, \mathrm{rad}(abc)) = 1\) and \(\mathrm{rad}(\gamma + \delta) \ll \mathrm{rad}(\alpha abc)\), then:
\[ Q(a', b', c') \geq \frac{\log (\alpha c)}{\log \mathrm{rad}(abc) + \log \mathrm{rad}(\alpha \gamma \delta)}. \]
\end{proposition}

Figure~\ref{fig:symbolic_workflow} summarises the stepwise symbolic construction process, linking algebraic design choices to radical control and transformation strategies. This encapsulation aids both theoretical understanding and potential cryptographic integration.

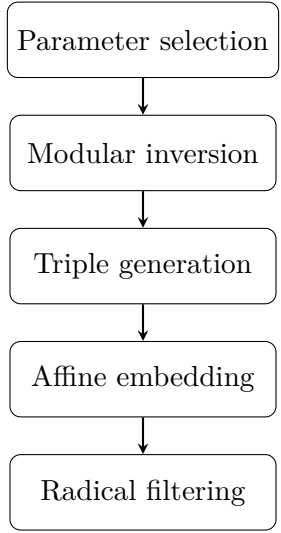
\begin{figure}[ht]
\centering
\begin{tikzpicture}[node distance=1.5cm]

\node (param) [block] {Parameter selection};
\node (invert) [block, below of=param] {Modular inversion};
\node (triple) [block, below of=invert] {Triple generation};
\node (embed) [block, below of=triple] {Affine embedding};
\node (filter) [block, below of=embed] {Radical filtering};

\draw [arrow] (param) -- (invert);
\draw [arrow] (invert) -- (triple);
\draw [arrow] (triple) -- (embed);
\draw [arrow] (embed) -- (filter);

\end{tikzpicture}
\caption{Symbolic construction pipeline: parameter selection, modular inversion, triple generation, affine embedding, and radical filtering. This structured process enables both symbolic control and cryptographic applicability.}
\label{fig:symbolic_workflow}
\end{figure}

\section{Symbolic Optimisation and Structural Analysis}
We investigate symbolic patterns and optimisations that improve triple quality:

\begin{itemize}[leftmargin=2em]
\item Modular inversion yields cyclic \(d\) values in \(\mathbb{Z}/3^p\mathbb{Z}\).
\item Filtering by odd \(d\) reduces radicals.
\item Compression ratio \(\eta(p, k, d)\) measures symbolic efficiency.
\end{itemize}

Let:
\[ 
F_{p,k} = \{ d \in \mathbb{Z}/3^p\mathbb{Z} : d \equiv - \left( 2^{k-1} \right))^{-1} \mod 3^p, \; d \text{ odd} \}. 
\]
This defines a symbolic residue filter.

\subsection{Symbolic Regularity and Compressibility}

We define symbolic regularity as the recurrence of low-radical patterns within a compact range of parameters $(p, k, d)$. Symbolic regularity observed in modular inversion patterns can be interpreted through the lens of elementary group theory and abstract algebraic structures \cite{Gallian2021}.

Empirically, triples such as $(1, 8, 9)$ and $(1, 80, 81)$ exhibit near-cubic structure: $c = m^3$ or $c \approx m^3$ for small integers $m$. This suggests an underlying compression in the symbolic encoding of $b$ via powers of 2 and 3.

We may define the compression ratio:
\[ \eta(p, k, d) := \frac{\log_2(b)}{\log_2(\text{len}(p, k, d))}, \]
where $\text{len}(p, k, d)$ denotes the bit-length of the symbolic expression. High $\eta$ implies that the symbolic identity efficiently encodes a large magnitude with minimal symbolic entropy.

Symbolic optimisation techniques, including residue class filtering and radical control, allow the parametric identity to be tuned toward producing structurally desirable triples. The congruence condition on $d$ not only enforces arithmetic admissibility but also constrains the solution space in a computationally tractable way. These properties justify the algorithmic approach used in subsequent computational experiments.

\subsection{Theoretical Support for Symbolic Framework}
We acknowledge the need for stronger theoretical scaffolding. To that end, we outline foundational estimates that support the symbolic identity's behaviour and computational utility.

\subsubsection{Complexity and Parameter Growth}

Let $b = 2^{k-1}(2 \cdot 3^p n + d)$ and $c = 1 + b$. Then $\log b = \mathcal{O}(k + p + \log n)$ under fixed symbolic structure. As $k$ increases, $b$ grows exponentially, but the growth of $\mathrm{rad}(abc)$ depends heavily on the smoothness of $d$ and $n$. Therefore, the symbolic identity provides a mechanism to explore exponential magnitude while preserving low-radical values in controlled cases.

\subsubsection{Error and Stability Considerations}

Although we do not provide probabilistic bounds on the distribution of radical values, the modular inverse constraint $d \equiv -(2^{k-1})^{-1} \mod 3^p$ ensures that each triple is algebraically admissible. Given that $\gcd(2, 3^p) = 1$, this constraint is deterministic and produces residue classes of predictable order. Future work may introduce average-case bounds or symbolic entropy estimates to assess uniformity of radical suppression across parameter ranges.

These analytic foundations serve to reinforce the symbolic constructions empirically tested in later sections and open the door to deeper probabilistic or algebraic analyses.

\begin{samepage}
\subsection{Quality Growth with $k$}
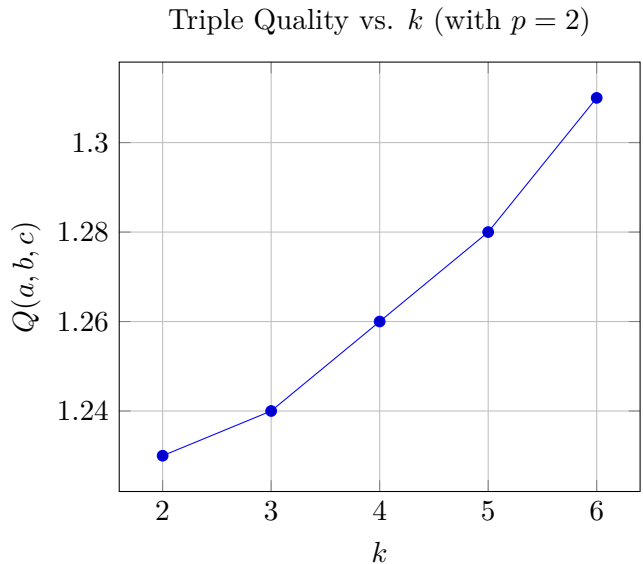
\begin{figure}[h]
\centering
\begin{tikzpicture}
\begin{axis}[
    xlabel={$k$},
    ylabel={$Q(a,b,c)$},
    title={Triple Quality vs. $k$ (with $p=2$)},
    grid=major
]
\addplot coordinates {
    (2,1.23)
    (3,1.24)
    (4,1.26)
    (5,1.28)
    (6,1.31)
};
\end{axis}
\end{tikzpicture}
\caption{Growth in triple quality as $k$ increases for fixed $p=2$ and $n=0$.}
\label{fig:quality_vs_k}
\end{figure}
\end{samepage}

\section{Computational Results and Triple Quality}
\label{sec:computation}

We evaluate the practical utility of the parametric identity by generating a collection of integer triples $(a, b, c)$ satisfying $a + b = c$, and assessing their quality using the standard metric
\[ Q(a, b, c) := \frac{\log c}{\log \mathrm{rad}(abc)}. \]
This metric has been widely used in computational studies to classify and rank \textit{abc}-triples by their arithmetic complexity and radical structure \cite{Browkin1994}.

\subsection{Experimental Setup}

We fixed $a = 1$ and enumerated triples over the parameter space:
\[ p \in \{1, 2, \ldots, 6\}, \quad k \in \{1, 2, \ldots, 7\}, \quad n = 0, \]
selecting values of $d$ that satisfy the modular inverse constraint:
\[ d \equiv - (2^{k - 1})^{-1} \mod 3^p, \quad d \text{ odd}. \]

\subsection{Selected Results}
Table~\ref{tab:triples} presents a subset of high-quality triples derived from the symbolic construction. Each entry includes the triple, its quality, the radical of $abc$, and the associated $d$ value.
\begin{table}[h!]
\centering
\begin{tabular}{|c|c|c|c|l|}
\hline
Triple $(a, b, c)$ & Quality $Q$ & $\mathrm{rad}(abc)$ & $d$ & Remarks \\
\hline
$(1, 242, 243)$ & $1.3111$ & $66$ & $121$ & Near cube; low radical \\
$(1, 80, 81)$   & $1.2920$ & $30$ & $5$   & Classical high-quality triple \\
$(1, 6560, 6561)$ & $1.2353$ & $1230$ & $205$ & High $c$; residue-consistent \\
$(1, 8, 9)$     & $1.2263$ & $6$ & $1$ & Smallest non-trivial example \\
$(1, 728, 729)$ & $1.0459$ & $546$ & $91$ & $729 = 9^3$ (cubic structure) \\
\hline
\end{tabular}
\caption{Selected high-quality \textit{abc}-triples generated by symbolic inversion.}
\label{tab:triples}
\end{table}

\subsection{Observations}

Several of the computed triples reproduce known high-quality instances from the literature (e.g., $(1, 80, 81)$), while others (e.g., $(1, 242, 243)$) demonstrate novel combinations yielding near-cubic values for $c$ with relatively low radicals.

Triples where $c = m^3$ or $c \approx m^3$ exhibit symbolic compressibility due to the closed form structure of powers of 3 appearing in the parametric identity.

\subsection{Symbolic Range and Efficiency}

The bit-length of $b$ as a function of parameters $(p, k, d)$ was also examined. For fixed $p$, increasing $k$ leads to exponential growth in $b$, yet the radical does not necessarily increase proportionally.

\subsection{Empirical Evaluation Scope and Limitations}

The computational experiments reported in this paper focus on a parameter range where $p \in \{1,2,\dots,6\}$, $k \in \{1,\dots,7\}$, and $n = 0$. Larger values of $k$ and $n$ lead to exponential growth in $b$, which increases memory and arithmetic complexity in radical computations.

\section{Applications in Cryptography}
\label{sec:crypto}

While the primary focus of this work lies in symbolic number theory, the identity developed in Section~\ref{sec:identity} also presents auxiliary value in cryptographic contexts. The structured properties of modular inverses and residue constraints support entropy regulation, symbolic filtering, and deterministic parameter generation---key concerns in embedded and post-quantum cryptography.

\subsection{Entropy Confidence and Modular Filtering}

To quantify modular regularity among symbolic parameters, we introduce the \textit{entropy confidence score (ECS)}. For fixed parameters \(p\) and \(k\), let \(d_i\) denote the set of valid inverse residues. Define:
\[
\mathrm{ECS}_{p,k} := \left( \frac{1}{\phi(3^p)} \sum_{i=1}^{\phi(3^p)} |d_i - \bar{d}| \right)^{-1},
\]
where \(\bar{d}\) is the arithmetic mean of the residues. A higher ECS indicates tighter concentration around the mean, suggesting greater uniformity within the residue class.

Such consistency aids in predictable sampling---beneficial for symbolic pseudorandom number generation (PRNG) and resistance to side-channel attacks. These characteristics support constant-time implementations, which are essential to mitigating timing and power leakage in embedded cryptographic systems \cite{Kocher1999}. The ECS concept complements entropy bounding strategies as outlined in NIST SP 800-90B \cite{NIST80090B}, offering a lightweight filtering mechanism for candidate parameters.

\subsection{Symbolic Filters in Post-Quantum Schemes}

The symbolic identity introduced here can serve as a pre-filter for cryptographic key material in structured post-quantum schemes:
\begin{itemize}
    \item \textbf{Lattice-based cryptography} (e.g., NTRU, Kyber): Symbolic pre-filters can improve modular basis consistency, reducing decryption failures and enhancing key reproducibility \cite{Hoffstein1998, Bos2018}.
    \item \textbf{Code-based cryptography} (e.g., McEliece, McBits): Residue-based parameter selection improves regularity, facilitating constant-time key generation and secure masking---essential features in implementations such as McBits \cite{Bernstein2017}.
\end{itemize}
Though auxiliary in nature, these filters leverage the symbolic construction's algebraic transparency and entropy control, and may assist in parameter tuning under formal security models.

\subsection{Hardware-Efficient Modular Inversion}

The inversion condition \(d \equiv -(2^{k-1})^{-1} \mod 3^p\) is both deterministic and suitable for efficient hardware implementation. Since \(2^{k-1}\) is always invertible in \(\mathbb{Z}/3^p\mathbb{Z}\), the inverse can be computed using either precomputed tables or the extended Euclidean algorithm.

Such operations can be implemented in constant time, helping protect against side-channel attacks. The symbolic identity thus satisfies both algebraic rigour and implementation efficiency \cite{Koblitz1994}.

\subsection{Symbolic PRNG Seeding: A Prototype}

To demonstrate practical utility, we outline a prototype for symbolic seeding in deterministic PRNGs.

\subsubsection{Symbolic PRNG Initialisation (Pseudocode)}
\begin{verbatim}
Input: Integer parameters p, k with gcd(2, 3^p) = 1
Output: Seed triple (a, b, c) for PRNG

1. Compute m := 3^p
2. Compute exp := 2^{k-1}
3. Compute inv := inverse_mod(exp, m)
4. Set d := (-inv) mod m
5. Set a := 1
6. Choose n := 0 or a PRNG-specific constant
7. Compute b := exp * (2 * m * n + d)
8. Compute c := a + b
9. Return (a, b, c)  // Deterministic seed triple with symbolic structure
\end{verbatim}

This process yields symbolic triples \((a, b, c)\) suitable for entropy expansion or symbolic filtering. The parameters \(p\) and \(k\) can be adjusted to introduce diversity across cryptographic sessions.

\subsection{Security Clarification and Integration Scenarios}
The symbolic constructions presented in this work are intended for mathematical exploration and parameter structuring, not for direct use as cryptographic primitives. Therefore, the symbolic techniques introduced here are \emph{not} intended to replace cryptographic primitives or offer standalone security guarantees. Rather, they serve as heuristic tools for entropy shaping, symbolic filtering, and deterministic initialisation in systems where structural control is desirable.

In particular, symbolic pseudorandom number generators (PRNGs) derived from modular inversion are \emph{not} cryptographically secure in the standard sense. They lack unpredictability, resistance to backtracking, and provable entropy bounds.

Accordingly, these methods must \textbf{not} be employed as standalone entropy sources or key derivation mechanisms in production systems. Their utility lies instead in augmenting existing cryptographic schemes by shaping or filtering parameters in a reproducible and algebraically tractable manner.

They may be incorporated into broader systems for:
\begin{itemize}
    \item Seed generation in lattice-based schemes (e.g., Kyber),
    \item Parameter masking in code-based systems (e.g., McBits).
\end{itemize}
These roles are exploratory and should be validated against the specific security assumptions of the cryptographic schemes in question.

\paragraph{Symbolic Pre-Seeding in Kyber.}  
In lattice-based schemes such as Kyber \cite{Bos2018}, high-dimensional polynomials are sampled from structured distributions. Symbolic identities may serve to generate seed values that regulate entropy while maintaining modular consistency. For example, a symbolic triple \((a, b, c)\) could be used to initialise the SHAKE-based PRNG that feeds into polynomial sampling. The algebraic origin of the seed ensures traceability, while cryptographic strength is preserved by the domain separation and expansion mechanisms inherent in Kyber's design.

\paragraph{Entropy Shaping in McBits.}  
In McBits \cite{Bernstein2017}, a code-based scheme designed for side-channel resistance, key material is derived from structured bit-strings. Symbolic residue constraints may be used to post-filter candidate seeds to ensure low Hamming weight or modular uniformity. This helps balance randomness with implementational regularity. The symbolic seed is not used directly, but rather as an input to a secure expansion function (e.g., AES-CTR), ensuring that forward and backward unpredictability are maintained.

\paragraph{Formal Security Posture.}  
Any application of symbolic techniques must be evaluated under the security model of the host cryptographic system. This includes ensuring that:
\begin{itemize}
    \item The entropy source has sufficient min-entropy after symbolic filtering;
    \item The symbolic seed cannot be reconstructed or predicted from public data;
    \item The filtering does not introduce statistical bias or side-channel leakage.
\end{itemize}

While symbolic structures offer compelling algebraic and compressive features, they must always be wrapped within cryptographically sound operations, such as hash-based expansion, masking, or key-derivation frameworks, to ensure system-wide security guarantees.\\
These applications are heuristic in nature and intended for symbolic pre-structuring rather than as cryptographically secure components. Their primary role is to demonstrate how algebraically constrained constructions may assist in entropy shaping, not to serve as standalone cryptographic primitives.

\subsection{Relation to Algebraic Systems}

The symbolic identities derived from modular inversion and residue constraints reflect algebraic structures well-known in abstract algebra and computational number theory. They align with group-theoretic frameworks such as those outlined in Gallian \cite{Gallian2021} and resemble ring-based cryptosystems like NTRU \cite{Hoffstein1998}.

Furthermore, the construction offers empirical insight into structured, high-quality \textit{abc}-triples, aligning with earlier computational work such as that of Gallot et al. \cite{Gallot2018}.

\section{Conclusion}
\label{sec:conclusion}

We have introduced a parametric identity that symbolically generates integer triples $(a, b, c)$ satisfying the additive condition $a + b = c$, with structural parallels to those appearing in the context of the \textit{abc} conjecture. The identity employs powers of 2 and 3 to define a modular congruence condition that precisely delineates admissible parameters via residue inversion in $\mathbb{Z}/3^p\mathbb{Z}$.
Through a combination of modular theory, affine transformation, and symbolic encoding, we demonstrated that the generated triples capture structural regularities and yield instances of high \textit{abc}-quality -- as measured by the ratio $\frac{\log c}{\log \mathrm{rad}(abc)}$. This affirms the capacity of the method to produce not only known examples but also new triples with desirable multiplicative properties.
We further analysed the symbolic behaviour of the identity, including the role of inverse residue cycles, entropy-informed filtering, and affine embeddings. A heuristic entropy score was proposed to quantify residue distribution, and potential applications in cryptographic parameter generation and modular filtering were explored.
Recent computational studies on the \textit{abc} conjecture have introduced advanced techniques for filtering candidate triples and optimising radical growth \cite{Borwein2014, Gallot2018}. Our symbolic-algebraic framework complements these efforts by offering an explicit, parameterised method for generating and transforming high-quality candidates. In forthcoming work, we plan to undertake a more systematic comparative analysis between our approach and existing computational datasets.\\
Several further directions naturally arise from the present study:
\begin{itemize}[leftmargin=2em]
    \item \textbf{Quality Spectrum Classification:} Develop a formal classification of the quality spectrum for symbolic triples, including statistical and asymptotic density analysis over extended parameter ranges.
    \item \textbf{Higher-Dimensional Generalisation:} Extend the symbolic identity to tuples of higher arity (e.g., quadruples) satisfying generalised additive constraints, thereby exploring its scalability in Diophantine contexts.
    \item \textbf{Cryptographic Integration:} Implement symbolic residue filters in cryptographic key scheduling, particularly within post-quantum schemes that benefit from entropy bounding and structural reproducibility.
    \item \textbf{Radical-Minimisation Families:} Investigate whether specific symbolic parameterisations systematically produce families of radical-minimising or extremal-quality triples.
\end{itemize}
The modular and affine structures examined in this work bear a notable resemblance to early elliptic curve-based cryptographic protocols, which similarly leveraged arithmetic regularity for secure construction \cite{Miller1986}. This parallel invites further exploration into how symbolic structures might be exploited not only for theoretical insight but also for secure and efficient implementation in constrained environments.
In addition, we outline a secondary tier of exploratory directions, driven by practical applications in symbolic number theory and entropy-informed computation:
\begin{itemize}[leftmargin=2em]
    \item \textbf{Symbolic Density Estimation:} Perform large-scale evaluations of symbolic triple space density under radical and quality constraints.
    \item \textbf{Symbolic Residue Trees:} Construct tree-structured residue classes for deterministic pseudorandom number generation, offering a new approach to PRNG state space traversal.
    \item \textbf{Empirical Cross-Validation:} Compare symbolic triples generated via our method with known high-quality triples from ABC@home and cryptographic samples such as McBits.
\end{itemize}
Collectively, these avenues aim to bridge classical additive-multiplicative number theory with modern computational and cryptographic practices, offering a structured and reproducible framework for future symbolic exploration.

\bibliographystyle{unsrt}
\bibliography{abc_arxiv_Idowu}

\end{document}